\documentclass{endm}
\usepackage{endmmacro}
\usepackage{graphicx,amsmath}
\usepackage{enumerate}

\begin{document}

\begin{verbatim}\end{verbatim}\vspace{2.5cm}

\begin{frontmatter}

\title{Proper Hamiltonian Cycles in Edge-Colored Multigraphs}

\author[ESP1]{Raquel \'Agueda}
\author[LRI]{Valentin Borozan}
\author[ESP2]{Raquel D\'iaz}
\author[LRI]{Yannis Manoussakis} 
\author[LRI]{Leandro Montero}

\address[ESP1]{Departamento de An\'alisis Econ\'omico y Finanzas, Universidad de Castilla, \\ 
La Mancha 45071 Toledo, Spain. \\{\rm \texttt{raquel.agueda@uclm.es}}}

\address[LRI]{Laboratoire de Recherche en Informatique, Universit\'e Paris-Sud, \\ 
91405 Orsay CEDEX, France. \\{\rm \texttt{\{valik,yannis,lmontero\}@lri.fr}}}

\address[ESP2]{Facultad de Ciencias Matem\'aticas, Universidad Complutense de Madrid, \\ 
28011 Madrid, Spain. \\{\rm \texttt{radiaz@ucm.es}}}

\begin{abstract} 
A $c$-edge-colored multigraph has each edge colored with one of the $c$ available colors and no two parallel edges
have the same color. A proper Hamiltonian cycle is a cycle containing all the vertices 
of the multigraph such that no two adjacent edges have the same color. In this work we establish sufficient conditions for a multigraph to have
a proper Hamiltonian cycle, depending on several parameters such as the number of edges and the rainbow degree.
\end{abstract}


\end{frontmatter}

\section{Introduction}\label{intro}
The study of problems modeled by edge-colored graphs gave place to important
developments over the last years. For instance, the research on long colored cycles
and paths for edge-colored graphs has given interesting results. We refer to~\cite{Bang-Jensen2001} for a survey on such results. 
From the point of view of applicability, problems arising in molecular biology are often modeled using colored
graphs, i.e., graphs with colored edges and/or vertices~\cite{Pevzner2000}. Given such an edge-colored
graph, original problems translate to extracting subgraphs colored in a specified pattern.
The most natural pattern in such a context is that of a proper coloring, i.e., adjacent edges
have different colors. In this work we give sufficient conditions involving various parameters such as the number of edges and the rainbow degree,
in order to guarantee the existence of properly colored Hamiltonian cycles in edge-colored multigraphs. 
We note that the proper Hamiltonian path and proper Hamiltonian cycle problems are both $NP$-complete in the general case.
However, a proper Hamiltonian path can be determined in polynomial time in $c$-edge-colored complete graphs for $c \geq 2$~\cite{Feng}.
It is also polynomial to find a proper Hamiltonian cycle in $c$-edge-colored complete graphs for $c = 2$, see~\cite{bankfalvi},
but it is still open to determine the computational complexity for $c\geq 3$~\cite{Benkouar1996}.
Many other results for edge-colored multigraphs can be found in the survey by Bang-Jensen and Gutin~\cite{BG97}.
Results involving only degree conditions can be found in~\cite{Abouelaoualim2010}.

Formally, let $I_c=\{1,2,\ldots, c\}$ be a set of $c \geq 2$ colors. Throughout this paper $G^c$ denotes a \emph{c-edge-colored 
multigraph} such that each edge is colored with one color in $I_c$ and no two edges joining the same pair of vertices 
have the same color. Let $V(G^c)$ and $E(G^c)$ be the vertex set and the edge set of $G^c$ respectively. 
Set $n=|V(G^c)|$ and $m=|E(G^c)|$. If $H$ is either a set of vertices or a subgraph of $G^c$, and $x$ is a vertex of $G^c$, then $N^i_H(x)$ denotes
the set of vertices of $H$ adjacent to $x$ with an edge of color $i$ and $d^i_H(x)$ denotes its cardinality. 
Whenever $H$ contains all the vertices of $G^c$, 
we write $N^i(x)$ instead of $N^i_{G^c}(x)$. The \emph{colored i-degree} of a vertex $x$,
denoted by $d^i(x)$, is the cardinality of $N^i(x)$. 
The \emph{rainbow degree} of a vertex $x$, denoted by $rd(x)$, 
is the number of different colors on the edges incident to $x$. The \emph{rainbow degree} of a multigraph $G^c$, denoted by $rd(G^c)$, is 
the minimum rainbow degree among its vertices. An edge with endpoints $x$ and $y$ is denoted by 
$xy$, and by $c(xy)$ we denote the set of colors present on the edges between $x$ and $y$. The \emph{rainbow complete multigraph} is the one having all possible colored edges between any pair of vertices
(its number of edges is therefore $c\binom{n}{2}$). The \emph{complement} of a multigraph $G^c$ denoted by $\overline{G^c}$, is 
a multigraph with the same vertices as $G^c$ and an edge $vw \in E(\overline{G^c})$ on colour $i$ if and only if $vw \notin E(G^c)$ on that colour. 
We say that an edge $xy$ is a \emph{missing edge} of $G^c$ if $xy \in E(\overline{G^c})$.
A subgraph of $G^c$ is said to be \emph{properly colored} or just \emph{proper}, if any two adjacent edges in this subgraph differ 
in color. A \emph{path} is a sequence of different vertices $v_1v_2\ldots v_k$ such that $v_i$ is adjacent to $v_{i+1}$ for $i=1,\ldots,k-1$.
For $k\geq 3$, if $v_k$ is adjacent to $v_1$ then  $v_1v_2\ldots v_kv_1$ is called a \emph{cycle}.
A \emph{Hamiltonian path} (\emph{cycle}) is a path (cycle) containing all vertices of the multigraph. 
A \emph{proper Hamiltonian path} (\emph{cycle}) is a Hamiltonian path (cycle) that is properly colored.
All multigraphs are assumed to be connected. Throughout this paper we will abbreviate colors $red$, $blue$ and $green$ with $r$, $b$ and $g$.
For example, we will write $d^r(x)$ instead of $d^{red}(x)$.

The following results for edge-colored multigraphs will be useful in our work.

\begin{theorem}[\cite{Abouelaoualim2010}]\label{Abou}
Let $G^c$ be a $2$-edge-colored multigraph on $n$ vertices colored with $\{r,b\}$. If for every vertex $x$ we have that $d^r(x)\geq \left\lceil \frac {n+1}{2} \right\rceil $ and
$d^b(x)\geq \left\lceil \frac {n+1}{2}\right\rceil $, then $G^c$ has a proper Hamiltonian cycle for $n$ even and a proper cycle of length $n-1$ for $n$ odd. 
\end{theorem}

\begin{theorem}[\cite{paperPHP}]\label{3colorsPHP} Let $G^c$ be a $c$-edge-colored
multigraph on $n$ vertices, $n \geq 2$ and  $c\geq 3$. If $m\geq
c\binom{n-1}{2}+1$, then $G^c$ has a proper Hamiltonian path.
\end{theorem} 

The next lemma is the corresponding result for proper Hamiltonian cycles to the one for proper Hamiltonian paths presented in~\cite{paperPHP}. 
It reduces the case $c\geq 4$ just to the case $c=3$. We omit the proof since it is exactly the same as for paths.

\begin{lemma}\label{to3colorscycle} 
Let $\ell$ be a positive integer. Let $G^c$ be a $c$-edge-colored connected
multigraph on $n$ vertices and $m \geq c \ \ell + 1$ edges, where $c \geq 4$. There
exists a color $c_j$ such that if we color the edges of $G^{c_j}$ with another color
and we delete parallel edges with the same color, then the resulting
$(c-1)$-edge-colored multigraph $G^{c-1}$ is connected and has $m' \geq (c-1) \ell + 1$
edges. Furthermore, if $G^{c-1}$ has a proper Hamiltonian cycle then $G^c$ has one
too. Also, if $rd(G^c)=c$, then $rd(G^{c-1})=c-1$.
\end{lemma} 

This paper is organized as follows: In Section~\ref{2_edge_col} 
we study proper Hamiltonian cycles in $2$-edge-colored multigraphs and in Section~\ref{3_edge_col} 
we study proper Hamiltonian cycles in $c$-edge-colored multigraphs, for $c \geq 3$. This division is because of different bounds and proofs. 

\section{Bounds for 2-edge-colored multigraphs}\label{2_edge_col}

In this section we study the existence of proper Hamiltonian cycles in $2$-edge-colored 
multigraphs. We present two main results. The first one involves the number of edges. The second one, the rainbow degree and 
the number of edges. Both results are tight.

The lemma below will be useful for Theorem~\ref{s1}

\begin{lemma}\label{s0} Let $G^c$ be a $2$-edge-colored multigraph colored with $\{r,b\}$.
Assume that $G^c$ contains a proper cycle $C$ of length at most $n-2$, and that there exists a red edge $xy$ in $G^c-C$.  If
$d^b_C(x)+d^b_C(y)> |C| $, then $G^c$ has a proper cycle of length $|C|+2$
containing $xy$.  
\end{lemma} 
\begin{proof} Set $C=x_1y_1x_2y_2\ldots x_sy_sx_1$, where $x_iy_i$ are the blue edges of $C$, $i=1,2,\ldots ,s$. Then
$d^b_{\{x_i, y_i\}}(x)+d^b_{\{x_i,y_i\}}(y) \leq 2$. Otherwise, if $d^b_{\{x_i,
y_i\}}(x)+d^b_{\{x_i,y_i\}}(y)$ $\geq 3$, then either $x_1y_1x_2y_2\ldots x_i x y y_i \ldots x_sy_sx_1$ 
or $x_1y_1x_2y_2\ldots x_i y x y_i \ldots x_sy_sx_1$ is the desired cycle. It follows that $\sum_i d^b_{\{
x_i, y_i\}}(x)+d^b_{\{x_i,y_i\}}(y) \leq 2 \frac {|C|}{2}=|C|$, a contradiction
to the hypothesis of the lemma. This completes the proof.  
\end{proof}

Now we can prove the following theorem.

\begin{theorem}\label{s1} 
Let $G^c$ be a $2$-edge-colored multigraph on $n$ vertices, $n \geq 4$ colored with $\{r,b\}$.
If $m\geq 2\binom{n-1}{2}+n$, then $G^c$ has a proper
Hamiltonian cycle if $n$ is even, and a proper cycle of length $n-1$ otherwise.
\end{theorem} 

For the extremal example, consider a
rainbow complete $2$-edge-colored multigraph on $n-1$ vertices ($n$ even). Add
one new vertex $x$. Then add all possible edges in the same color between $x$ and the complete multigraph. 
Clearly, the resulting multigraph has $2\binom{n-1}{2}+n-1$ edges and it has no proper Hamiltonian cycle since $x$ 
has just a single incident color.

\begin{proof} The proof is by induction on $n$. The theorem is clearly true for small values of $n$, say
$n=4,5$. Let us suppose that $n\geq 6$. By Theorem~\ref{Abou}, if
for every vertex $x$ we have that $d^r(x)\geq \left\lceil \frac {n+1}{2} \right\rceil $ and
$d^b(x)\geq \left\lceil \frac {n+1}{2}\right\rceil $, then $G^c$ has a proper
Hamiltonian cycle for $n$ even and a proper cycle of length $n-1$ for $n$ odd.
Let us suppose therefore that for some vertex, say $x$, and for some color, say
red, $d^r(x) \leq \left\lceil \frac {n+1}{2}\right\rceil -1$. Notice now that
$d^r(x)>0$ and $d^b(x)>0$, otherwise, for example if $d^r(x)=0$, then $m\leq
2\binom{n}{2}-(n-1)<2\binom{n-1}{2}+n$, a contradiction. Similarly, if $d^r(x)+d^b(x)\leq 2$, then $m\leq
2\binom{n}{2}-2(n-2)<2\binom{n-1}{2}+n$, again a contradiction. Thus we may conclude
that there are two distinct neighbors, say $y$ and $z$, of $x$ such that
$r \in c(xy)$ and $b \in c(xz)$ in $G^c$. Replace the vertices $x,y,z$ by a new
vertex $s$ such that $N^b(s)= N^b_{G^c-\{x,z\}}(y)$ and  $N^r(s)=
N^r_{G^c-\{x,y\}}(z)$. The obtained multigraph, say $G'^c$, has $n-2$ vertices and
at least  $2\binom{n-1}{2}+n -[(n-1)+\left\lceil \frac {n+1}{2}
\right\rceil-1+2(n-2)] > 2\binom{n-2}{2}+n-1$ edges, i.e., the number of edges needed for the inductive hypothesis 
in a multigraph on $n-2$ vertices. So, for $n-2$ even, $G'^c$ has a proper Hamiltonian cycle and 
then coming back to $G^c$ we may easily find a proper Hamiltonian cycle in
$G^c$. Assume now that $n-2$ (and thus $n$) is odd. Then, $G'^c$ has a proper cycle, say $C$,
of length $n-3$. If $s$ belongs to $C$, then as before we may
easily find a proper cycle of length $n-1$ in $G^c$. Assume therefore that $s$
does not belong to $C$. By Lemma~\ref{s0},  if $d^b_C(x)+d^b_C(y)> |C|$ (or
$d^r_C (x)+d^r_C(z)>|C|$) then we may integrate the edge $xy$ (respectively
$xz$) in $C$ in order to obtain a cycle of length $n-1$. Otherwise, 
$d^b_C (x)+d^b_C (y)\leq |C| $ and  $d^r_C (x)+d^r_C(z)\leq |C|$. But then the
number of edges of $G^c$ is at most $2\binom{n}{2}-(n-3)-(n-3)<2\binom{n-1}{2}+n$, again a contradiction. 
This completes the argument and the proof.  
\end{proof}

\begin{theorem}\label{2colrd2} 
Let $G^c$ be a $2$-edge-colored multigraph on $n$ vertices, $n \geq 9$ colored with $\{r,b\}$. 
If $rd(G^c)=2$ and $m\geq \binom{n}{2}+\binom{n-2}{2}+3$, then $G^c$ has a proper
Hamiltonian cycle if $n$ is even, and a proper cycle of length $n-1$ otherwise.
\end{theorem}

For the extremal multigraph, consider a complete blue graph, say $A$, on $n-2$ vertices, $n$ even. Add two
new vertices $v_1,v_2$ and join them to a vertex $v$ in $A$ with blue
edges.  Finally, superimpose the obtained graph with a complete red graph on $n$
vertices. Although the resulting $2$-edge-colored multigraph has rainbow degree two and $\binom{n}{2}+\binom{n-2}{2}+2$ edges, 
it has no proper Hamiltonian cycle. For this, observe that there is not a perfect blue matching
because $v_1$ and $v_2$ are only adjacent in blue to $v$.

\begin{proof}
The proof is by induction on $n$. For $n=9,10$, the result holds by inspection.
Suppose now that $n\geq 11$. Observe that $|E(\overline{G^c})| \leq 2n-6$. By Theorem~\ref{Abou}, if
for every vertex $v \in G^c$ we have that $d^r(v)\geq \left\lceil \frac {n+1}{2} \right\rceil $ and
$d^b(v)\geq \left\lceil \frac {n+1}{2}\right\rceil $, then $G^c$ has a proper
Hamiltonian cycle if $n$ is even, and a proper cycle of length $n-1$ otherwise. Suppose then that there exists a vertex $v$ such that 
$d^r(v)\leq \lceil \frac {n+1}{2}\rceil -1$. 

Now $v$ has two different neighbors $u$ and $w$ such that $b \in c(vu)$ and $r \in c(vw)$. Otherwise, if $v$ has just one neighbor with edges on both 
colors, then $|E(\overline{G^c})| = 2n-4 > 2n-6$, a contradiction. We construct a new multigraph $G'^c$ by replacing the vertices $v,u$ and $w$ with a 
new vertex $z$ such that $N^r(z)=N_{G^c-\{v,w\}}^r(u)$ and $N^b(z)=N_{G^c-\{v,u\}}^b(w)$. 
We will show that either $rd(G'^c)=2$ and we can apply the inductive hypothesis, or there are two vertices in $G^c$ with $r-$ or $b-$colored degree at 
most two and thus, we can find the desired cycle for $G^c$ directly.

Assume first that $rd(G'^c)=2$. We see that $G'^c$ has at least
$\binom{n}{2}+\binom{n-2}{2}+3 - (n-1) - (\lceil \frac {n+1}{2}\rceil -1) - (n-3) - (n-3) -2$ edges. This number, for $n \geq 11$, is greater or equal than 
$\binom{n-2}{2}+\binom{n-4}{2}+3$, i.e., the number of edges needed to have a proper Hamiltonian cycle or a proper cycle of length $n-3$ ($n$ odd) in $G'^c$.
So by the inductive hypothesis we obtain such a cycle. For $n$ even, it is easy to obtain a proper Hamiltonian cycle in $G^c$, since we deleted the 
appropriate edges at $u$ and $w$. For $n$ odd, if the new vertex $z$ is on the cycle of length $n-3$, it is exactly the same as for the even case to 
obtain a proper cycle of length $n-1$ in $G^c$. Now, if the vertex $z$ is not on the proper cycle in $G'^c$, then neither are the 
vertices $v,u,w$ in $G^c$. Let $x_1y_1x_2y_2\ldots x_ky_kx_1$ be
the proper cycle, for $2k = n-3$. Suppose without losing generality that the edges $x_iy_i$ are red and the edges $y_ix_{i+1}$ are blue. 
If we can add neither the blue edge $vu$ nor the red edge $vw$ to the cycle, then we have at most two red edges between the endpoints of the edge 
$vu$ and the endpoints of the edges $x_iy_i$ and at most two blue edges between the endpoints of the edge $vw$ and the endpoints of the edges $y_ix_{i+1}$.
So, since the length of the cycle is $n-3$, there are $\frac{n-3}{2} + \frac{n-3}{2} = 2n-6$ edges in $\overline{G^c}$. Therefore, we have all possible 
edges in $G^c$ except those missing ones. In particular we have the red edge $uv$ and the red edge $uw$. Now, if there is a blue edge from $v$ to some vertex 
$x_i$ ($y_i$) of the cycle, we extend the proper cycle to a proper cycle of length $n-1$ with the red edge $uv$ adding the blue edge $vx_i$ ($vy_i$), 
the blue edge $uy_{i-1}$ ($ux_{i+1}$) and removing the blue edge $x_iy_{i-1}$ ($y_ix_{i+1}$) from the cycle. The blue edge $uy_{i-1}$ ($ux_{i+1}$) clearly exists 
since $u$ has all possible blue incident edges. If there is no blue edge from $v$ to some vertex $x_i$ ($y_i$) of the cycle, we have that there exists 
a blue edge from $w$ to some vertex $x_i$ ($y_i$) of the cycle. Otherwise, we have that $d^b(v)= 2$ and $d^b(w)= 2$, and we have covered that case before. 
Finally, we extend the proper cycle to a proper one of length $n-1$ exactly as we did to add the red edge $uv$ but now with the red edge $uw$.

Suppose last that $rd(G'^c)<2$. Then there exists a vertex $x_{G'^c}$ with $rd(x_{G'^c}) < 2$ in $G'^c$. 
Therefore in $G^c$ the vertex corresponding to $x_{G'^c}$, say $x$, has either $d^r(x)\leq 2$ or $d^b(x)\leq 2$, and it is clearly different 
from $v$. We can assume that $d^r(v)\leq 2$ (similar if $d^b(v)\leq 2$) and it follows that $v$ and $x$ have colored degree at most two in $G^c$. 
Otherwise, take another neighbor of $x$, say $y$, such that the edges $xv$, $xy$ have different colors. Then 
we can replace in $G^c$ the vertices $x,v,y$ with a new vertex $z$ as described before, obtaining a new multigraph $G''^c$.
In this case, if $rd(G''^c)=2$ we apply the previous case. If not, there is a vertex $x'_{G''^c}$ in $G''^c$ with $rd(x'_{G''^c}) < 2$. 
Thus in $G^c$ the vertex corresponding to $x'_{G''^c}$, say $x'$, has either $d^r(x')\leq 2$ or $d^b(x')\leq 2$ and also $x' \neq x$. 
Therefore we could obtain $x$ and $x'$ with colored degree at most two in $G^c$.
Observe now that as $d^r(v)\leq 2$ and, either $d^r(x)\leq 2$ or $d^b(x)\leq 2$, then $E(\overline{G^c}) = 2n-6$, and this happens when 
$d^r(v)= 2$ and, either $d^r(x)= 2$ or $d^b(x)= 2$. 
Otherwise we have a contradiction with the assumption on the number of edges. So $G^c$ must have all possible edges except those $2n-6$ edges 
already missing at $v$ and $x$. We have the following cases now depending on the vertex $x$ in $G^c$.
\begin{enumerate}[a)]
 \item $x=u$, therefore $d^r(u)= 2$, that is, there is a red edge between $u,v$ and a red edge between $u,w$.  
Since $v$ and $u$ have all possible blue incident edges, take a neighbor of $v$, say $v' \neq w$, such that $b \in c(vv')$ and a neighbor of $w$,
say $w'\neq v'$, such that $b \in c(ww')$. Now, as $G^c-\{v,u\}$ is rainbow complete take a proper Hamiltonian path, for $n$ even, or a proper path
of length $n-3$, for $n$ odd, from $v'$ to $w'$ such that the color of its first edge is the same as its last edge, that is, color $r$. 
Call this path $P$. Then in $G^c$, $vv'Pw'wv$ is a proper Hamiltonian cycle if $n$ is even and a proper cycle of length $n-1$ otherwise. 
 \item $x=w$, therefore $d^b(w)= 2$, that is, there is a blue edge between $w,v$ and a blue edge between $w,u$. This case is similar to the previous one, just taking 
a proper Hamiltonian path, for $n$ even, or a proper path of length $n-3$, for $n$ odd, in $G^c-\{v,w\}$ from $u$ to a neighbor $v'$ of $v$
such that $b \in c(vv')$ and the color of the first edge of the path is the same as its last edge, that is, color $r$. 
If we call this path $P$, then in $G^c$, $vv'Puwv$ is a proper Hamiltonian cycle if $n$ is even and a proper cycle of length $n-1$ otherwise.
 \item $x\neq u$ and $x \neq w$. Suppose first that $d^r(x)=2$. Then there is a red edge between $x,v$ and a red edge between $x,w$. In this case we proceed exactly as 
in the case (a) adding the red edge $vx$ to the proper Hamiltonian path (or proper path of length $n-3$) that exists in $G^c-\{v,x\}$. 
Suppose now that $d^b(x)=2$. Then there is a blue edge between $x,v$ and a blue edge between $x,u$. Here, we repeat the argument of case (b), 
adding the blue edge $vx$ to the proper Hamiltonian path (or proper path of length $n-3$) that exists in $G^c-\{v,x\}$.
\end{enumerate}

The proof is now complete.
\end{proof}

\section{Bounds for c-edge-colored multigraphs, $c\geq 3$}\label{3_edge_col}

In this section we study the existence of proper Hamiltonian cycles in $c$-edge-colored
multigraphs, $c\geq 3$. We present two main results. The first one involves the number of edges. The second involves the rainbow degree and 
the number of edges. Both results are tight. Finally, we state a conjecture involving the rainbow degree, the number of edges and the connectivity.

\begin{theorem}\label{3colgen} 
Let $G^c$ be a $c$-edge-colored multigraph on $n$ vertices, $n \geq 4$ and $3 \leq c < n$. 
If $m\geq c\binom{n-1}{2}+n$, then $G^c$ has a proper Hamiltonian cycle. 
\end{theorem}

For the extremal example, consider a rainbow complete $c$-edge-colored multigraph on $n-1$ vertices. Add
one new vertex $x$.  Then add all possible edges in one color, say red, between $x$ and the complete multigraph. 
Clearly, the resulting multigraph has $c\binom{n-1}{2}+n-1$ edges and it has no proper Hamiltonian cycle since $x$ 
has only edges on color red incident to it.

\begin{proof} Suppose first $c\geq 4$ and that the conclusion of Lemma~\ref{to3colorscycle} is not sufficient in order to consider just the case $c=3$.
Let $E^j$ denote the edges of $G^c$ on color $j$ for $j=1,\ldots,c$. Let $mono(E^j)$ be the set of edges as in $E^j$ but in the underlying simple graph $G$. 
Now it follows for each pair $j$, $l$ of distinct colors that $|mono(E^j) \cap mono(E^l)| \geq \binom{n-1}{2}+1$. 
That is, there are at least $\binom{n-1}{2}+1$ parallel edges on colours $j$ and $l$.
Now take any two colors $j,l$ and consider the uncolored simple graph $G$ having same vertex set 
as $G^c$ and for each pair of vertices $x,y$, we add the uncolored edge $xy$ in $G$ if and only if $xy \in E^j$ and $xy \in E^l$ in $G^c$. 
Clearly $G$ has at least $\binom{n-1}{2}+1$ edges. We distinguish between two cases depending on the connectivity of $G$.

Suppose first that $G$ is not $2$-connected. Therefore $G$ has exactly $\binom{n-1}{2}+1$ edges, that is, $G$ is isomorphic to a 
complete graph on $n-1$ vertices, say $B$, plus one vertex, say $v$, adjacent to only one vertex of $B$. 
Consequently in $G^c$, $v$ has two different neighbors, say $x$ and $y$, such that the edges $vx$, $vy$ have different colors and these colors are not $j$ or $l$.
Now as $B$ is a complete graph, there exists a Hamiltonian path from $x$ to $y$, that is,
an alternating path $P$ on colors $j$, $l$ in $G^c$. Then $vxPyv$ is a proper Hamiltonian cycle in $G^c$.

Suppose next that $G$ is $2$-connected. Now as $G$ has at least $\binom{n-1}{2}+1$ edges and it is $2$-connected, 
by~\cite{ByerSmeltzerDM2007}, $G$ contains a Hamiltonian cycle $C=v_1\ldots v_nv_1$. If $n$ is even, then $C$ is a proper Hamiltonian cycle in $G^c$ 
using alternatingly colors $j$ and $l$. Suppose therefore that $n$ is odd. Let $k$ be a color in $G^c$ different from $j$ and $l$.
As $|mono(E^j) \cap mono(E^k)| \geq \binom{n-1}{2}+1$, there are at most $n-2$ pairs of vertices $x$ and $y$ in $G^c$ such that 
the edge $xy$ on color either $j$ or $k$, is missing.
However, as the cycle $C$ has $n$ edges, there exist two consecutive vertices of $C$ joined by an edge on color $k$. 
Suppose without loss of generality that $k \in c(v_1v_n)$. We can conclude that $v_1\ldots v_nv_1$ is a proper Hamiltonian cycle in $G^c$ where 
$v_1\ldots v_n$ is an alternating path on colors $j$, $l$ and the edge $v_nv_1$ is on color $k$.

Assume therefore that $c=3$ and $G^c$ has the edges colored with $\{r,b,g\}$. So we have that $m\geq 3\binom{n-1}{2}+n$. 
We prove the theorem by induction on $n$. For $n=4$ the theorem is easily checked. Suppose then that 
$n \geq 5$. Consider two vertices $v$ and $w$ such that there are three parallel edges between them on colors $r,b$ and $g$ respectively. 
It can be checked that such two vertices always exist, otherwise if there are at most two edges between every pair of vertices, then
the number of edges would be at most $2\binom{n}{2}$ contradicting the hypothesis. 

Suppose first that either $d(v) \geq 3n-4$ or $d(w) \geq 3n-4$, say $d(v) \geq 3n-4$. Then, by removing $v$ from $G^c$ we obtain a multigraph with 
at least $3\binom{n-1}{2}+n-3(n-1)=3\binom{n-2}{2}+n-3 \geq 3\binom{n-2}{2}+1$ edges. Therefore by Theorem~\ref{3colorsPHP}, 
$G^c-\{v\}$ contains a proper Hamiltonian path. Then, since $d(v) \geq 3n-4$ we can connect $v$ to the endpoints of the path in an appropriate way 
in order to obtain a proper Hamiltonian cycle in $G^c$. 

Suppose next that $d(v) \leq 3n-5$ and $d(w) \leq 3n-5$. Now consider a new multigraph $G'^c$ obtained from $G^c$ by contracting $v$ and $w$ to a single 
vertex $z$ such that $N^r(z) = N^r(v)-\{w\}, N^b(z) = N^b(w)-\{v\}$ and $N^g(z) = N^g(v) \cap N^g(w)$. Within this contraction we removed at most $3n-5$ edges 
from $G^c$; thus $G'^c$ has at least $3\binom{n-1}{2}+n -(3n-5)\geq 3\binom{n-2}{2}+n-1$ edges. By induction, $G'^c$ has a proper Hamiltonian cycle 
and then we can easily obtain one in $G^c$. 
\end{proof}

\begin{theorem}\label{3colrd3} 
Let $G^c$ be a $c$-edge-colored multigraph on $n$ vertices, $n \geq 4$ and  $c\geq 3$. 
If $rd(G^c)=c$ and $m\geq c\binom{n-1}{2}+c+1$, then $G^c$ has a proper Hamiltonian cycle.
\end{theorem}

For the extremal case, consider a
rainbow complete $c$-edge-colored multigraph on $n-1$ vertices. Add
one new vertex $x$. Then, add all possible edges in all colors between $x$ and one vertex of the complete multigraph. 
Clearly, the resulting multigraph has $c\binom{n-1}{2}+c$ edges, rainbow degree $c$ but it has no proper Hamiltonian cycle since it is 
not $2$-connected.

For the proof of Theorem~\ref{3colrd3} we need the following lemma.

\begin{lemma}\label{lemma3colrd3} 
Let $G^c$ be a $3$-edge-colored multigraph on $n \geq 4$ vertices colored with $\{r,b,g\}$ fulfilling the conditions of Theorem~\ref{3colrd3}.
If there exists a vertex $x$ such that $d^i(x)=1$ for some $i \in \{r,b,g\}$, then $G^c$ has a proper Hamiltonian cycle.
\end{lemma}
\begin{proof}
We can suppose without loss of generality that $d^r(x)=1$. Let $z$ be the only neighbor of $x$ such that $r \in c(xz)$ and let $y\neq z$ be 
another neighbor of $x$ such that $b \in c(xy)$. Clearly this vertex $y$ exists in $G^c$, otherwise $G^c$ would have at most
$3\binom{n}{2} - 3(n-2)$ edges (as $x$ could only be adjacent to $z$ with edges on colors $r,b,g$) and this is less than $3\binom{n-1}{2}+3+1$,
that is, a contradiction on the number of edges.
Now, contract the vertices $x,y$ and $z$ to a new vertex $w$ such that 
$N^r(w) = N_{G^c-\{x,z\}}^r(y), N^b(w) = N_{G^c-\{x,y\}}^b(z)$ and $N^g(w) = N_{G^c-\{x,y\}}^g(z) \cap N_{G^c-\{x,z\}}^g(y)$.
Within this contraction we removed at most $5n-7$ edges from $G^c$.
Let $G'^c$ be the obtained multigraph on $n-2$ vertices. 
We will show that either $G'^c$ has at least $3\binom{n-1}{2}+3+1 -(5n-9)$ edges (thus we can apply Theorem~\ref{3colgen}) or 
$G^c$ has a proper Hamiltonian cycle. In both cases the result holds.

Note that we can assume that there is at least one edge between $x$ and $z$ with color $b$ or $g$, otherwise if both edges are not present, 
$G'^c$ would have at least $3\binom{n-1}{2}+3+1-(5n-9)$ edges, as desired. 
Assume in the sequel that we have at least two parallel edges between $x$ and $z$, one on color $r$ and the other on color $b$ or $g$, say $b$.
Now consider the multigraph $G^c-\{x,z\}$. We can check by Theorem~\ref{3colorsPHP} that $G^c-\{x,z\}$ has a proper Hamiltonian path 
$P=v_1v_2\ldots v_{n-2}$. We can also assume that $y=v_1$ (or $y=v_{n-2}$); otherwise, if there is no edge between $x$ and any of $v_1,v_{n-2}$ 
on colors $b$ or $g$, we get to delete two edges fewer, leaving at most $5n-9$ edges to be deleted in the contraction. We have the following cases now.

Suppose first that there are exactly two edges on colors $r$ and $b$ between $x$ and $z$. 
So, we need to find one edge less in order to have at most $5n-9$ edges to delete. Therefore, $x$ should be adjacent to both $y=v_1$ and $v_{n-2}$ 
with parallel edges on colors $b$ and $g$, otherwise, we obtain the last edge as required. Now, if $z$ is adjacent to $v_1$ with edges 
on both colors $b$ and $g$, then $xzv_1\ldots v_{n-2}x$ is a proper Hamiltonian cycle in $G^c$. 
Otherwise, we obtain one edge less to delete and we are done.

Suppose last that there are three edges on colors $r,b$ and $g$ between $x$ and $z$. 
Assume $x$ is adjacent to $v_1$ with only one edge on color $b$. Then there are two parallel edges between $x$ and $v_{n-2}$ 
on colors $b$ and $g$, otherwise we obtain two more missing edges. Now, if $z$ is adjacent to $v_1$ with two parallel edges 
on both colors $b$ and $g$, then $xzv_1\ldots v_{n-2}x$ is a proper Hamiltonian cycle in $G^c$. Otherwise, we obtain no more than 
$5n-9$ deleted edges. Assume finally that there exist two parallel edges between $x$ and $v_1$ on colors $b$ and $g$. 
Now, if the edge $xv_{n-2}$ exists on a color different from $v_{n-3}v_{n-2}$, then
either $xzv_1\ldots v_{n-2}x$ is a proper Hamiltonian cycle in $G^c$ or there are two missing edges between $z$ and $v_1$, as desired.
Otherwise, there exists only one edge between $x$ and $v_{n-2}$ where $xv_{n-2}$, $v_{n-3}v_{n-2}$ have the same color.
So, if we have an edge $zv_{n-2}$ on a color different from $v_{n-3}v_{n-2}$, then 
$zxv_1\ldots v_{n-2}z$ is a proper Hamiltonian cycle in $G^c$. Otherwise, there are two missing edges between $z$ and $v_{n-2}$.
Therefore if we take $x,z,v_{n-2}$ for the contraction instead of $x,z,y$, then we remove $5n-9$ edges.
This completes the argument of this case and the proof.
\end{proof}

\noindent \textbf{Proof of Theorem~\ref{3colrd3}.}
By Lemma~\ref{to3colorscycle} it is enough to prove the theorem for $c=3$. Assume that the edges of $G^c$ are colored with $\{r,b,g\}$. 
The proof is by induction on $n$. For $n=4,5$ the theorem is easily checked. 
Suppose then that $n \geq 6$. 

Suppose that for every vertex $v$ in $G^c$ we have $d(v) \geq 3n-6$.
Then $m \geq \frac{1}{2}\sum_v (3n-6) = \frac{3n^2-6n}{2} \geq 3\binom{n-1}{2}+n = \frac{3n^2-7n+6}{2}$ for $n\geq 6$. 
Therefore by Theorem~\ref{3colgen} we obtain a proper Hamiltonian cycle in $G^c$. 

We can assume next that there exists a vertex $v$ in $G^c$ such that $d(v) \leq 3n-7$.
Suppose now that there exists a vertex $w$ in $G^c$ such that $d(w) \geq 3n-4$. Clearly $v \neq w$.
Now we remove $w$ from $G^c$ obtaining a multigraph with at least $3\binom{n-1}{2}+4 - 3(n-1)=3\binom{n-2}{2}+1$ edges. 
By Theorem~\ref{3colorsPHP}, $G^c-\{w\}$ has a proper Hamiltonian path. Since $d(w) \geq 3n-4$ (i.e. at most one edge is missing at $w$), 
we can easily join $w$ to the endpoints of this proper path in order to obtain a proper Hamiltonian cycle in $G^c$. 
In what follows assume that for every vertex $w\neq v$ in $G^c$ we have $d(w) \leq 3n-5$.
Now we distinguish between the following two cases depending on the neighbors of $v$.

\begin{itemize}
 \item \emph{$v$ has a neighbor $w$ such that there are three parallel edges between $v$ and $w$ on colors $r,b$ and $g$ respectively.}
Consider a new multigraph $G'^c$ obtained from $G^c$ by contracting $v$ and $w$ to a single 
vertex $z$ such that $N^r(z) = N^r(v)-\{w\}, N^b(z) = N^b(w)-\{v\}$ and $N^g(z) = N^g(v) \cap N^g(w)$. In order to apply induction on $G'^c$, 
we need to show that it has at least $3\binom{n-2}{2}+4$ edges and $rd(G'^c)=3$.  Observe that $3\binom{n-1}{2}+4-(3\binom{n-2}{2}+4)=3n-6$, that is, 
the maximum number of edges that we are allowed to delete in the contraction process. Indeed, as $d(v) \leq 3n-7$ and $d(w) \leq 3n-5$, 
we can choose the colors to delete in order to remove $3n-6$ edges as desired. Now, if $rd(G'^c)=3$, then by induction, $G'^c$ has a proper 
Hamiltonian cycle and this cycle can be easily transformed to a proper one in $G^c$. 
Otherwise suppose that $rd(G'^c)\leq 2$, that is, there exists a vertex $x$ in $G'^c$ such that $rd(x) \leq 2$. 
If $x\neq z$, then it is easy to see that there exists a color $i \in \{r,b,g\}$ such that $d_{G^c}^i(x) = 1$.
Thus, by Lemma~\ref{lemma3colrd3} the result holds.

Assume now that $x=z$. If $z$ has no edge incident on color $r$ (or by symmetry on color $b$), then we have that $d_{G^c}^r(v) = 1$ 
and again by Lemma~\ref{lemma3colrd3} we obtain a proper Hamiltonian cycle in $G^c$.
If $z$ has no edge incident on color $g$, then we conclude that $d_{G^c}^g(v) + d_{G^c}^g(w) \leq n$.
As colors $r,b$ and $g$ are symmetric, from all above arguments we conclude that 
$d_{G^c}^b(v) + d_{G^c}^b(w) \leq n$ and $d_{G^c}^r(v) + d_{G^c}^r(w) \leq n$.
Consider now the multigraph $G'^c= G^c-\{v,w\}$. $G'^c$ has $n-2$ vertices and at least $3\binom{n-1}{2}+4 - 3n = 3\binom{n-2}{2}-2$ edges, i.e.,
$G'^c$ is almost rainbow complete (as only two edges are missing). Take a neighbor, say $v'$, of $v$ in $G'^c$ and a neighbor, say $w'\neq v'$, of $w$
in $G'^c$. Now we take any proper Hamiltonian path $P=v_1\ldots v_{n-2}$ in $G'^c$ such that $v'=v_1$ and $w'=v_{n-2}$, and the pair of edges $vv'$, $v'v_2$ and 
$ww'$, $w'v_{n-3}$ have different colors respectively. Observe that this path always exists since $G'^c$ is almost rainbow complete. 
Finally $vv'Pw'wv$ is a proper Hamiltonian cycle in $G^c$ as the edge $vw$ has the three colors to choose in order to take one 
different from the colors on the edges $vv'$ and $ww'$.

 \item \emph{$v$ has no neighbor $w$ such that there are three parallel edges between $v$ and $w$ on colors $r,b$ and $g$.}
This implies that $d(v) \leq 2n-2$. Consider two distinct neighbors of $v$, say $x$ and $y$, such that the edges $vx$, $vy$ have different colors.
As in the proof of Lemma~\ref{lemma3colrd3}, contract $v,x$ and $y$ to a new vertex $z$ and let $G'^c$ be the obtained multigraph. 
If within this contraction we removed $5n-9$ edges from $G^c$, then $G'^c$ has a proper Hamiltonian cycle by Theorem~\ref{3colgen} and 
therefore one in $G^c$. Observe now that, as $d(v) \leq 2n-2$, if we do the contraction we remove at most $5n-8$ edges. Therefore, if 
$d(v) \leq 2n-3$ we are finished. We can suppose then that $d(v)=2n-2$, that is, there exist exactly two parallel edges between $v$ and every other 
vertex. Consider some neighbor $x$ of $v$. Now as $d(x) \leq 3n-5$ and there is one edge missing between $v$ and $x$, there 
exists a neighbor $y$ of $x$ such that there are at most two parallel edges between them. Since there are 
exactly two parallel edges between $v$ and $y$, if we contract these three vertices $v,x$ and $y$ choosing appropriately which colors to delete 
we remove at most $5n-9$ edges as desired.
\qed
\end{itemize}

To finish this section, we state a conjecture for the existence of proper Hamiltonian cycles depending not only on rainbow degree and the number of 
edges but also on the connectivity.

\begin{conjecture}\label{3colrd3conn} 
Let $G^c$ be a $2$-connected $c$-edge-colored multigraph on $n$ vertices, $n \geq 10$ and $c\geq 3$.
If $rd(G^c)=c$ and $m\geq c\binom{n-2}{2}+4c+1$, then $G^c$ has a proper Hamiltonian cycle.
\end{conjecture}

If true, Conjecture~\ref{3colrd3conn} is tight. For this, consider a rainbow complete $c$-edge-colored multigraph on $n-2$ vertices. 
Add two new vertices $x_1$ and $x_2$. Consider two vertices of the complete multigraph $y_1$ and $y_2$. Finally, add all possible edges in all colors 
between $\{x_1,x_2\}$ and $\{y_1,y_2\}$. The resulting multigraph has $c\binom{n-2}{2}+4c$ edges, rainbow degree $c$ and it is $2$-connected. 
However, this multigraph has no proper Hamiltonian cycle since the vertices $x_1, x_2$ cannot both be together in such a proper Hamiltonian cycle.

\bibliography{biblio}

\end{document}